\documentclass[10pt]{article}%
\usepackage{amsfonts}
\usepackage{amsmath}
\usepackage{amssymb}
\usepackage{amscd}
\usepackage{graphicx}
\usepackage{amsthm}%

\usepackage{latexsym}
\usepackage{amsfonts}
\usepackage{amssymb}
\usepackage{amscd}
\usepackage{amsmath}
\newtheorem{definition}{Definition}[section] 
\newtheorem{theorem}{Theorem}[section] 
\newtheorem{proposition}{Proposition}[section] 
 
\newtheorem{example}{Example}[section] 
\usepackage[letterpaper,margin=2cm]{geometry}

\usepackage{color}
\definecolor{red}{rgb}{1,0.2,0.2}
\definecolor{green}{rgb}{0.2,1,0.5}
\definecolor{blue}{rgb}{0,0,1}
\definecolor{lightblue}{rgb}{0.3,0.5,1}

\newenvironment{remark}{{\it Remark. }}{\hfill $\triangleleft$ \medskip}

\def\avg{
\begin{array}{c}
{} \\
{\rm avg} \\
{x^0\in {\cal S}}
\end{array} }

\def\j2from{
\begin{array}{c}
{j=1 } \\
{j\neq i}
\end{array} }
\def\[{\begin{equation}}
\def\]{\end{equation}}
\title{\bf A Note on Local Mode-in-State Participation Factors for Nonlinear Systems }

\author{Boumediene Hamzi and Eyad H. Abed
\thanks{Boumediene Hamzi is with the Department of Mathematics, AlFaisal University, Riyadh, KSA. {\tt\small boumediene.hamzi@gmail.com}; Eyad H. Abed is with the Department of Electrical and Computer Engineering and the Institute for Systems Research, 
        University of Maryland, College Park, MD 20742, USA. 
        {\tt\small abed@umd.edu}. Parts of this work were done when the first author was a Marie Curie fellow at Imperial College London (London, UK) then Ko\c{c} University (Istanbul, Turkey).}%
}
\def\NN{\mbox{I}\!\mbox{N}}
\newcommand{\RR}{\mathbb{R}}
\begin{document}

\maketitle
\thispagestyle{empty}
%\pagestyle{empty}
%herehere removed \pagestyle{empty} above

\begin{abstract}
The paper studies an extension to nonlinear systems of a recently proposed approach to the concept of modal participation factors. First, a definition is given for local mode-in-state participation factors for smooth nonlinear autonomous systems. The definition is general, and, unlike in the more traditional approach, the resulting participation measures depend on the assumed uncertainty law governing the system initial condition. The work follows Hashlamoun, Hassouneh and Abed (2009) in taking a mathematical expectation (or set-theoretic average) of a modal contribution measure with respect to an assumed uncertain initial state. As in the linear case, it is found that a symmetry assumption on the distribution of the initial state results in a tractable calculation and an explicit and simple formula for mode-in-state participation factors. 
%The dichotomy discovered in the linear case between mode-in-state participation factors and state-in-mode participation factors is encountered again in the nonlinear setting.
\end{abstract}

\section{Introduction}
Analysis of modal content of the response of dynamic systems is of interest in many application areas, ranging from electric power networks to vibration of structures. Many approaches to modal analysis occur in the literature. For linear time invariant systems, modal content consists of the eigenmodes, and can be studied analytically. For nonlinear systems, the possibility of global oscillations gives rise to global oscillatory modes that might not be connected to the eigenmodes of the system's linearization at an equilibrium. In this paper, we focus on local modal analysis of nonlinear autonomous systems near an equilibrium, paying particular attention to what can be viewed as eigenmodes in a neighborhood of the equilibrium point of interest. The aim of the paper is to explore the possibility of extending to the nonlinear setting the modal participation analysis pursued by Hashlamoun, Hassouneh and Abed \cite{abed} for linear systems. This analysis attempts to systematically quantify the relative contributions of system modes to system states, and of system states to system modes. Here, system states refers to the scalar elements of the system state vector.

In the early 1980s, Verghese, Perez-Arriaga and Schweppe \cite{Perez-Arriaga1,Perez-Arriaga2} introduced quantities they referred to as modal participation factors. These quantities have been used widely, especially in the electric power systems field. In 2009, the authors of \cite{abed} presented a new approach to the fundamental definition of modal participation factors. The idea of modal participation factors, which will be reviewed further in the next section, is to give measures of the relative contribution of system modes in system states, and of system states in system modes. In \cite{abed}, such measures are developed by taking an average of relative contribution measures over an uncertain set of system initial conditions. The idea is that fixing the system initial condition affects the modal participations, and that initial conditions are in reality uncertain, indeed possibly random due to inherent noise. Indeed, if one takes a view that the initial time also isn't fixed, noise can be viewed as having the effect of allowing the initial condition to be re-set over time, effectively allowing the initial condition to explore a neighborhood of an equilibrium point over a short time interval. By taking an averaging approach, the authors of \cite{abed} find that a dichotomy arises in this new view of modal participation factors. In this dichotomy, participation factors measuring {\it mode-in-state participation} need to be viewed as distinct from participation factors measuring {\it state-in-mode participation}. This dichotomy was not recognized prior to \cite{abed}, and a single formula was previously used to quantify both types of modal participation.

In \cite{abed}, it was found that analytical formulas for mode-in-state participation factors fell out of the analysis very nicely, under basic symmetry assumptions on the distribution of the initial state. {The same symmetry assumptions did not allow for a similarly simple derivation of state-in-mode participation factors, and when a formula was obtained in a particular scenario on the initial state, that formula was more complicated than for the mode-in-state case and didn't share the desirable property of being independent to rescaling of the system state variables (i.e., the formula wasn't invariant under changes of state variable units). This issue is now better understood by the authors, and will be reported on elsewhere.}

Here, we explore extension of the work in \cite{abed}, {especially for the analysis of mode-in-state participation}, to the nonlinear setting, for local system behavior near an equilibrium point. We are able to give an analysis and derivation of formulas for mode-in-state participation factors (under basic symmetry assumptions as in the linear case). 
%Quantifying state-in-mode participation is also discussed, and similar observations are made as were given in \cite{abed} for the case of linear systems. Further comments are also given on the dichotomy which apply to both the nonlinear and linear settings. 
This work follows a different approach to defining participation factors for nonlinear models than that pursued in \cite{vittal}, where modal participation was studied from a fixed initial state using Taylor series methods.

{Before proceeding to the development of the paper, it is perhaps useful to  provide a brief discussion of studies on modal participation, addressing motivation of researchers on this topic, the various approaches taken in different disciplines, and applications that have been pursued.}

{The present work is motivated by the original work of Verghese, Perez-Arriaga and Schweppe \cite{Perez-Arriaga1,Perez-Arriaga2} that was mentioned above. The authors of  \cite{Perez-Arriaga1,Perez-Arriaga2} introduced their notion of modal participation factors as a tool to aid in modal analysis of large power grids, with benefits anticipated in tasks such as model order reduction and control design. Oscillatory modes are common in power systems, and it is important to have systematic tools for their analysis. Since power grids consist of interconnected areas and can cover large expenses of territory (indeed entire continents), engineers are naturally interested in obtaining reduced models that capture modes of special interest. The modal participation factors of \cite{Perez-Arriaga1,Perez-Arriaga2} were employed for this purpose, in an overarching framework that the authors referred to as Selective Modal Analysis (SMA) (a recent review of SMA is \cite{verghese_sma2013}). In addition, modal analysis in a power grid should provide tools for determining the best sites for insertion of actuators to control modes that may be troubling or dangerous, or for determining the best locations for placing measurement devices that allow system operators to monitor such modes in real time. An example of a recent application of the concepts in \cite{Perez-Arriaga1,Perez-Arriaga2} to power systems is \cite{setiada2018}, which focuses on power grids with significant levels of renewable generation . Early examples of work on actuator placement in power networks using the original modal participation concept include \cite{kundur1992,yakout1994}. Recent examples of modal participation studies in power systems motivated by the more recent approach of Hashlamoun, Hassouneh and Abed \cite{abed} include \cite{ramos,moraco,hill,netto}. The approach has also been applied in power electronics \cite{powerelectronics} and electromagnetic devices \cite{cenedese2016}.}

{The term ``modal participation factors" is also commonly used in the field of structural analysis, with applications in mechanical, aerospace, and civil engineering. The concept of modal participation factors introduced in electric power engineering in \cite{Perez-Arriaga1,Perez-Arriaga2} was developed independently of the notion used in structural analysis. Modal participation factors as studied in structural analysis have been used, for example, to study vibrations of tall buildings \cite{park} and rotorcraft dynamics \cite{prasad}. The concepts of modal participation factors in electric power engineering and in structural analysis are distinct. In the structural analysis framework, the focus has been on the impact of forcing functions on modal response. In contrast, in the electric power engineering concept, a large ostensibly autonomous dynamic system is considered (motivated by the driving power grid application). Bridging between these frameworks could be a fruitful area for future investigation. The two types of modal participation factors (in electric power engineering and related control theory literature, and in structural analysis) are not absolute by any means. These concepts are definitions deemed suitable for various purposes by their authors and employed over many years by practitioners in the respective fields. Later researchers have at times proposed modifications to address a perceived need for improvements. For example, in structural analysis, Chopra \cite{chopra96} introduced a new notion of modal participation factor aiming to make major improvements to the standard definition used in that field, including providing a more clear measure of modal contribution of an external forcing function and removing unit dependence from the standard notion. Similarly, the original concept of \cite{Perez-Arriaga1,Perez-Arriaga2} in electric power engineering has been revisited in \cite{abed}, as noted above. In the remainder of the paper, we focus on the modal participation factors concepts that have been used in electric power engineering, beginning with the work of \cite{Perez-Arriaga1,Perez-Arriaga2} and continuing with the work of \cite{abed}. The concepts from structural analysis were mentioned above to provide context for this work in the larger literature, but will not be addressed in the technical work in this paper.}

The remainder of the paper is organized as follows. In Section 2, needed background material is recalled. In Section 3, mode-in-state participation factors are defined for nonlinear systems in the vicinity of an equilibrium point, under a symmetry assumption on the uncertainty in the system initial condition. 
%In Section IV, the issue of performing similar calculations for state-in-mode participation analysis is discussed, and further insights on the dichotomy of modal participation factors are discussed that apply in both the nonlinear and linear settings. 
Conclusions and issues for further research are discussed in Section 4.
A preliminary version of this paper appeared in \cite{bh_abed_cdc}.
 
 \section{Background}
 
 {In this section, we give background material that is relevant to our investigation. In particular, we recall modal participation factors (the original definition as well as the more recent work as described above). We also sample some of the applications of modal participation factors, including very recent references to the literature.  Finally, we recall two fundamental theorems on local representations of nonlinear autonomous systems. As noted above, the remainder of the paper focuses on modal participation analysis as pursued in the electric power engineering literature and related work in control theory, and on extending the concepts given in \cite{abed} to a nonlinear setting.}

\subsection{Modal Participation Factors for Linear Systems: Original Definition \cite{Perez-Arriaga1, Perez-Arriaga2}}
Let $\Sigma_L$ denote the linear time-invariant system
 \[\label{sigma_l}\Sigma_L: \dot{x}=Ax\]
 where $x \in \RR^n$ and the state dynamics matrix $A \in \RR^{n \times n}$ has $n$ distinct eigenvalues $\lambda_i$, $i=1,\ldots,n$.
 
The system state $x(t)$ of course consists of a linear combination of exponential functions \[x^i(t) = e^{\lambda_i t} c^i,\]  where the vectors $c^i$ are determined by the system's initial condition $x(0)$. These functions are the \emph{system modes} and are useful in \emph{modal analysis of linear systems}. 

%\subsubsection{Participation Factors: Original Definition}

 Let $r^i$ be the right eigenvector of the matrix $A$ associated with eigenvalue $\lambda_i$, $i=1,\cdots,n$, and let $\ell^i$ be the left(row) eigenvector of $A$ associated with the eigenvalue $\lambda_i$, $i=1,\cdots,n$. The right and left eigenvectors are taken to satisfy the normalization
\[\ell^i r^j=\delta_{ij}, \]
where $\delta_{ij}=1$ if $i=j$ and $\delta_{ij}=0$ if $i \ne j$.
 
 Given a linear system $\dot{x}=Ax$ with initial condition $x(0)=x^0$, its solution can be written as 
 \[x(t)=e^{At}x^0=\sum_{i=1}^n (\ell^i x^0)e^{\lambda_i t} r^i.\]
The $k-$th state variable evolves according to 
 \[x_k(t)=(e^{At}x^0)_k=\sum_{i=1}^n (\ell^i x^0)e^{\lambda_i t} r^i_k\]

Using these facts and taking {two scenarios with rather special initial conditions}, Verghese, Perez-Arriaga and Schweppe \cite{Perez-Arriaga1,Perez-Arriaga2} motivated the following definition of quantities $p_{ki}$ which they named modal participation factors: 
\begin{eqnarray}
\label{eq:pki}
p_{ki}:= \ell_k^i r_k^i
\end{eqnarray}
Choosing the initial condition to be  $x^0=e_k$, the unit vector {along} the $k$-th coordinate axis, the authors of \cite{Perez-Arriaga1,Perez-Arriaga2} gave reasoning for considering the quantities  $p_{ki}$ as mode-in-state participation factors.  The scalars $p_{ki}$ are dimensionless. Next, employing a coordinate transformation to focus on the system modes and considering instead an initial condition $x^0 = r^i$, the right eigenvector corresponding to $\lambda_i$, the quantities  $p_{ki}$ were also given an interpretation as state-in-mode participation factors. Thus, it has been very common in papers and books using modal participation factor analysis to interchangeably refer to participation of modes in states and participation of states in modes, always using the same formula
 (\ref{eq:pki}) for both types of participation measure.

\subsection{Modal Participation Factors for Linear Systems: Recent Approach  (\cite{abed})}
{In \cite{abed}, it was argued that a deeper analysis of modal participation would not necessarily lead to identical measures for mode in state and state in mode participation factors. This issue continues to deserve the attention of the control, dynamics, and power systems research communities, largely because of the importance of modal analysis in many complex systems, in power engineering and in other application areas.}
{Simple examples were used in \cite{abed} to motivate the need for a new approach to defining modal participation factors. In fact, the examples indicated that it would be desirable to achieve definitions that gave different measures} for mode-in-state participation factors and state-in-mode participation factors. Indeed, the examples showed that, especially when quantifying the contribution of system states in system modes, the formula  (\ref{eq:pki}) could well fall short of giving an intuitively acceptable result. Thus, new  fundamental definitions were given based on averaging over the system initial condition, taken to be uncertain.

The linear system
\[\dot{x}=A x\]
usually represents the small perturbation dynamics near an equilibrium. The initial condition for such a perturbation is usually viewed as being an uncertain vector of small norm. In \cite{abed}, new definitions of mode-in-state and state-in-mode participation factors were given using deterministic {(i.e., set-theoretic)} and probabilistic uncertainty models for the initial condition.

\begin{definition}
\label{def-set}
In the set-theoretic formulation, the participation
factor measuring relative influence of the mode associated with
$\lambda_i$ on state component $x_k$ is
\begin{eqnarray}
p_{ki} &:=& \avg \frac{(\ell^i x^0) r^i_k}{x^0_k}
\label{eq:defmode-in-state}
\end{eqnarray}
whenever this quantity exists. Here, $x_k^0 = \sum_{i=1}^n (\ell^i x^0)
r^i_k$ is the value of $x_k(t)$ at $t = 0$, and ``${\rm avg}_{x^0t
\in {\cal S}}$'' is an operator that computes the average of a {scalar}
function over a set ${\cal S}\subset R^n$ (representing the set of
possible values of the initial condition $x^0$).
\end{definition}

With a probabilistic description of the uncertainty in the initial
condition $x^0$, the average in~(\ref{eq:defmode-in-state}) is
replaced {in \cite{abed}} by a mathematical expectation:

\begin{definition}
\label{def-prob}
The general formula for the
participation factor $p_{ki}$ measuring participation of mode $i$ in
state $x_k$ becomes
\begin{eqnarray}
p_{ki} &:=& E ~ \left\{ \frac{(\ell^i x^0) r^i_k}{x^0_k}\right\}
\label{eq:def-prob}
\end{eqnarray}
where the expectation is evaluated using some assumed joint
probability density function $f(x^0)$ for the initial condition
uncertainty. (Of course, this definition applies only when the
expectation exists.)
\end{definition}

In \cite{abed}, it was found that  both Definition~ \ref{def-set} (Eq. (\ref{eq:defmode-in-state})) and Definition~\ref{def-prob} (Eq. (\ref{eq:def-prob}))  lead to a simple result that agrees with Eq.~(\ref{eq:pki}) under a symmetry assumption on the uncertainty in the initial condition. In the set-theoretic definition, the symmetry assumption is that the
initial condition uncertainty set ${\cal S}$ is symmetric with
respect to each of the hyperplanes $\{x_k=0\}$, $k=1,\dots,n$.  In the probabilistic setting of Definition~\ref{def-prob}, the assumption is that the the initial condition
components $x^0_1,x^0_2,\dots,x^0_n$ are independent random variables with marginal density functions which are symmetric with respect to $x_k^0=0$, $k=1,2,\cdots,n$, or are jointly uniformly distributed over a sphere centered at the origin. Under either the set-theoretic or probabilistic symmetry assumption, it was found in \cite{abed} that
the same expression originally introduced by Perez-Arriaga, Verghese
and Schweppe~\cite{Perez-Arriaga1,Perez-Arriaga2} results as a measure of mode-in-state participation factors:
\begin{eqnarray}
p_{ki} &=& \ell^i_k r^i_k. \label{eq:pf-classical}
\end{eqnarray}

\subsection{State-in-Mode Participation Factors}
Hashlamoun, Hassouneh and Abed \cite{abed} also gave similar set-theoretic and probabilistic definitions for {state-in-mode} participation factors for linear systems. The calculations were found to be less straightforward than for the mode-in-state participation factors setting, even under {the same} symmetry assumption {on the initial condition as used in the mode-in-state participation factor calculation}. We will not recall the details of the development of state-in-mode participation factors for linear systems from \cite{abed}. We will simply recall from \cite{abed} the general definition and an associated result for the case of distinct real eigenvalues to have an idea of the nature of the results.

\begin{definition}
\label{def-modeinstate}
The participation factor of state $x_k$ in mode $i$ is
\[\pi_{ki} := \mbox{E}\bigg\{\frac{\ell_k^i x_k^0}{\sum_{j=1}^n (\ell_j^i x_j^0)}\bigg\}=\mbox{E}\bigg\{\frac{\ell_k^i x_k^0}{z_i^0}\bigg\},\]
whenever this expectation exists, where $z_i^0=z_i(0)=\ell^i x^0$, and where $z_i(t)$ is the $i^{th}$ system mode 
\[z_i(t)=e^{\lambda_i t}\ell^i x^0=e^{\lambda_i t} \sum_{j=1}^n (\ell_j^i x_j^0). \]
\end{definition}

It was shown in \cite{abed} that 
\begin{eqnarray*}
\pi_{ki}&=&  \mbox{E}\bigg\{\frac{\ell_k^i x_k^0}{\sum_{j=1}^n (\ell_j^i x_j^0)}\bigg\} \\ 
&=& \ell_k^i r_k^i + \sum_{j=1, j \ne i}^n \ell_k^i r_k^j E\bigg\{ \frac{z_j^0}{z_i^0}\bigg\}
\end{eqnarray*}

Note that the first term in the expression for $\pi_{ki}$ coincides with $p_{ki}$, the original participation factors formula. However, the second term does not vanish in general. This is true even when the components $x_1^0$, $x_2^0$,
$\cdots$, $x_n^0$ representing the initial conditions of the state are assumed to be independent. 
Assuming that the units of the state variables have been scaled to ensure that the probability
density function $f(x^0)$ is such that the components
$x^0_1,x^0_2,\dots,x^0_n$ are jointly uniformly distributed over the
unit sphere in $R^n$ centered at the origin, modal participation factors were evaluated in \cite{abed} using Definition \ref{def-modeinstate}, yielding the following explicit formula that is applicable under the foregoing uncertainty model for the system initial state.
%herehere

%herehere Result below is a result and shouldn't be called a definition
\begin{proposition} (\cite{abed})
Under the assumption that the initial condition has a uniform probability density on a sphere centered at the origin,  the participation factor of state $x_k$ in mode $i$ is 
\begin{eqnarray}
\pi_{ki} &=& \ell^i_k r^i_k ~+~ \sum_{j=1,~j\neq i}^{n} \ell^i_k r^j_k
\frac{l^j(\ell^i)^T}{\ell^i(\ell^i)^T}. \label{eq:st-in-mode-res}
\end{eqnarray}
\end{proposition}

 \subsection{Poincar\'e Linearization}
 Poincar\'e linearization is a well known technique for transforming an autonomous nonlinear system into a locally equivalent linear system via diffeomorphism. The technique is useful in this paper for extending the definitions of mode-in-state participation factors proposed in \cite{abed} to the nonlinear setting. In the following, we review the technique.
 
 Consider a nonlinear {system of ordinary differential equations} \[\label{nonlinear_ode} \dot{x}=f(x),\]
 where $x \in \RR^n$ and $f$ is an analytic vector field on $\RR^n$. Let $A=\frac{\partial f}{\partial x}|_{x=0}$ be the Jacobian of $f$ at the origin.
  
 \begin{definition}(\cite{arnold})
 Given a matrix $A \in \RR^{n \times n}$ with eigenvalues $\lambda_i$, $i=1,\cdots,n$, we say that the $n-$tuple $\lambda=(\lambda_1,\cdots,\lambda_n)$  is resonant if among the eigenvalues there exists a relation of the form
  \[(m,\lambda)=\sum_{k=1}^n m_k \lambda_k=\lambda_s, \]
 where $m=(m_1,\cdots,m_n)$, $m_k \ge 0$,  $\sum_k m_k \ge 2$. Such a relation is called a \emph{resonance}. 
 The number $|m|=\sum_{k=1}^n m_k$ is called the \emph{order} of the resonance.
 %herehere Changed $|m|=\sum_{k=1}^n m_k \lambda_k=\lambda_s$ to $|m|=\sum_{k=1}^n m_k$
 \end{definition}

 \begin{example} (\cite{arnold})
 The relation $\lambda_1=2 \lambda_2$ is a resonance of order $2$; the relation $2 \lambda_1= 3 \lambda_2$ is not a resonance; the relation $\lambda_1+\lambda_2=0$, or equivalently $\lambda_1=2 \lambda_1+\lambda_2$, is a resonance of order $3$.
 \end{example}
 
 \smallskip\noindent
 \begin{theorem}[Poincar\'e's Theorem \cite{arnold}] If the eigenvalues of the matrix $A$ are nonresonant, then the  nonlinear ODE 
 \[\dot{x}=Ax+O(|| x||^2 )\]
 can be reduced to the linear ODE
 \[\label{linear_ode}\dot{y}=Ay \]
 by a formal change of variable $x=y+\cdots$ (the dots denote series starting with terms of degree two or higher).
 \end{theorem}

If the $n-$tuple  $\lambda=(\lambda_1,\cdots,\lambda_n)$ is resonant, we will say that 
$$x^m:=x_1^{m_1}\cdots x_n^{m_n} e_s$$
is resonant if $\lambda_s=(m,\lambda)$, $|m| \ge 2$ with $e_i$ a vector in the eigenbasis of $A$ and $x_i$ are the coordinates with respect to the basis $e_i$. For example, for the resonance $\lambda_1=2 \lambda_2$, the unique resonant monomial is $x_2^2 e_1$. For the resonance $\lambda_1+\lambda_2=0$, all monomials $(x_1 x_2)^k x_s e_s$ are resonant \cite{arnold}.
 
  \begin{theorem}[Poincar\'e-Dulac Theorem \cite{arnold}]  If the eigenvalues of the matrix $A$ are resonant, then the  nonlinear ODE 
 \[\dot{x}=Ax+\cdots \]
 can be reduced to the ODE
 \[\dot{y}=Ay+w(y) \]
 by a formal change of variable $x=y+\cdots$ (the dots denote series starting with terms of degree two or higher), where all monomials in the series $w$ are resonant.
 \end{theorem}

 There are also several convergence results associated with Poincar\'e linearization, of which the following is the most well known.
 
 \smallskip\noindent
 \begin{theorem}[Poincar\'e-Siegel]
Suppose the eigenvalues $\{\lambda_i\}$, $i=1,\cdots,n$, of the linear part of an analytic vector field at an equilibrium point are nonresonant and either $\mbox{Re}( \lambda_i) >0$, $i=1,\cdots,n$ or $\mbox{Re}(\lambda_i) < 0$, $i=,\cdots,n$, or the $(\lambda_i)$ satisfy the Siegel condition, i.e. are such that there exists $C>0$ and $\nu$ such that for all $i=1,\cdots,n$ 
 \[|\lambda_i-(m,\lambda)| \ge \frac{C}{|m|^{\nu}} \]
 for all $m=(m_1,\cdots,m_n)$, where $(m_i)$ are nonnegative integers with $|m|=\sum_{i=1}^n m_i \ge 2$. Then the power series in Poincar\'e's Theorem  converges in some neighbourhood of the equilibrium point. 
 \end{theorem}
 
 \begin{remark} There are also some convergence results in the case of resonant eigenvalues; the reader is encouraged to consult \cite{arnold} for further details on Poincar\'e linearization.\end{remark}
 
 \subsection{Hartman-Grobman Theorem}
 Another very important result in the
local qualitative theory of nonlinear ordinary differential equations is the Hartman-Grobman Theorem, which
says that near a hyperbolic equilibrium point $x^e$, the nonlinear system (\ref{nonlinear_ode})
has the same qualitative structure as the linear system (\ref{linear_ode}).

 \begin{theorem}\cite{perko}
Let $E$ be an open subset
of $\RR^n$ containing the origin, let $f  \in {C}^1(E)$, and let $\phi_t$ be the flow of the
nonlinear system (\ref{nonlinear_ode}). Suppose that $f(0) = 0$ and that the matrix $A = Df(0)$
has no eigenvalue with zero real part. Then there exists a homeomorphism
$\varphi$ of an open set $U$ containing the origin onto an open set $V$ containing the origin such that for each $x^0 \in U$, there is an open interval  $I_0 \subset \RR$
containing zero such that for all $x^0 \in U$ and $t \in I_0$
\[\varphi \circ \phi_t(x^0) = e^{At}\varphi(x^0),\]
i.e., $\varphi$ maps trajectories of  (\ref{nonlinear_ode}) near the origin onto trajectories of (\ref{linear_ode}) near
the origin.
 \end{theorem}
 
\section{Mode-in-State Participation Factors for Nonlinear Systems}
Consider a nonlinear ODE
\[\label{nlsys}\dot{x}=f(x)\]
with $f \in {\cal C}(\RR^n;\RR^n)$, $f(0)=0$, and consider the Taylor expansion of $f$ around the origin
\[\label{taylor_exp}\dot{x}=Ax+\tilde{f}^{[2]}(x)+O(||x||^3)\]
where $A=\frac{\partial f}{\partial x}|_{x=0}$ and $\tilde{f}^{[2]}$ represents terms of order 2. We have the following result.

\begin{theorem}If the eigenvalues of $A$ are nonresonant (resp. satisfy one of the conditions of the Poincar\'e-Siegel Theorem) then there exists a diffeomorphism that formally (resp. analytically) transforms the nonlinear ODE (\ref{nlsys}) into a linear ODE. In this case, the mode-in-state participation factors of  (\ref{nlsys}) are the same as those of the linearized system $\dot{x}=Ax$.
\end{theorem}

\begin{proof}
First, we normalize $A$ using the change of coordinates
\[\label{chancor1} z=V^{-1}x,\]
where $V=[r^1\, r^2 \, \cdots r^n]$ represents the matrix of right eigenvectors of $A$. Under the change of coordinates (\ref{chancor1}) the ODE (\ref{taylor_exp}) becomes
\[\dot{z}=\Lambda z + V^{-1} \tilde{f}^{[2]}(V^{-1}z)+O(||z||^3) :=\Lambda z + f^{[2]}(z)+O(||z||^3)\]

Next, we normalize the higher order terms through the change of coordinates
\[ \label{chancor2} \tilde{z}=\phi(z)=z+\phi^{[2]}(z)+O(||z||^3)=z+z^T\left[\begin{array}{c} P_1\\ \vdots \\ P_n\end{array} \right] z+O(||z||^3)\]
where $\phi \in {\cal C}(\RR^n;\RR^n)$. Using Poincar\'e linearization, we know that if the eigenvalues of $A$ are nonresonant, then there is a formal change of coordinates $\phi$ such that the trajectories of (\ref{nlsys})  are locally diffeomorphic to the trajectories of 
\[\label{linsys} \dot{\tilde{z}}=\Lambda \tilde{z}\]

If $\Lambda=\mbox{diag}(\lambda_i)|_{i=1}^n$ , then
\[\tilde{z}(t)=e^{\Lambda t} \tilde{z}(0), \]
whose $i-$th component is 
\[\tilde{z}_i(t)=e^{\lambda_i t}\tilde{z}_i(0).\]
Using (\ref{chancor2}), we get $z(t)=\phi^{-1}(e^{\Lambda t}\phi(z^0))$, which can be rewritten as
\[z(t)=e^{\Lambda t}\phi(z^0)-\phi(z(0))^Te^{\Lambda^t t} \left[\begin{array}{c} P_1\\ \vdots \\ P_n\end{array} \right] e^{\Lambda t}\phi(z(0)) +O(||z||^3),\]
and
\[z_i(t)=e^{\lambda_i t}\phi_i(z^0)-\phi^T(z^0)e^{\Lambda^T t}P_ie^{\Lambda t}\phi(z^0)+\cdots \]

Using (\ref{chancor1}), we get
\[x_k(t)=\left[ \begin{array}{ccc} r^1 \cdots r^n\end{array} \right]_{\scriptsize \mbox{k-th row} }\left[\begin{array}{c}z_1\\ \vdots\\ z_n \end{array} \right]=\sum_i r_k^i z_i(t) \] 

\[=\sum_i r_k^i (e^{\lambda_i t}\phi_i(z^0)-\phi^T(z^0)P_i\phi(z^0))+\cdots\]

%\subsection{Specialization to the Linear Case} \label{lin_case_section} 

It is instructive to consider the linear case first. We set $P_i=0$ and the higher order terms are also set to zero in (\ref{chancor2}). This gives

\[x_k(t)=\sum_{i=1}^nr_k^i e^{\lambda_i t}\phi_i(z^0)=\sum_{i=1}^nr_k^ie^{\lambda_i t}\ell^i x^0 \]

Then the participation of the $e^{\lambda_i t}$ mode in the state $x_k(t)$ is \[p_{ki}:=\mbox{avg}\frac{e^{\lambda_i t}r_k^i\ell^i x^0}{x_k(t)}|_{t=0}=\ell_k^i r_k^i\] (agreeing, of course, with the previous calculation of \cite{abed} in the linear case \cite{abed}).

%\subsection{Nonlinear Case: $P_i \ne 0$}

Next, we consider the nonlinear setting, where we assume that $P_i \ne 0$.  The participation of $e^{\lambda_i t}$ in $x_k(t)$ is obtained using the set-theoretic definition as follows (quantities are evaluated at time $t=0$):
$$\mbox{avg}\frac{e^{\lambda_i t}r_k^i\phi_i(z^0)}{x_k(t)}|_{t=0}=\mbox{avg}\frac{e^{\lambda_i t}r_k^i\phi_i(z^0)}{\sum_{i=1}^nr_k^i(e^{\lambda_it}\phi_i(z^0)-\sum_{j,m}\theta_{j,m}e^{(\lambda_j+\lambda_m)t})}|_{t=0}$$

Since $\phi_i(z^0)=\ell^ix^0+\cdots$, then
\[\sum_{j,m}\theta_{j,m}=\sum_{j,m}\phi_j(z^0)\phi_m(z^0)p_{j,m}= \sum_{j,m} (\ell^jx^0)(\ell^mx^0)p_{j,m}\]

Hence, the participation of the mode $e^{\lambda_i t}$ in $x_k(t)$ is 
\begin{eqnarray*}p_{ki}&:=&\mbox{avg}\frac{e^{\lambda_i t}r_k^i\phi_i(z^0)}{x_k(t)}|_{t=0}\\ &=& \mbox{avg}\frac{e^{\lambda_i t}r_k^i\phi_i(z^0)}{\sum_{i=1}^nr_k^i(e^{\lambda_it}\phi_i(z^0)-\sum_{j,m}\theta_{j,m}e^{(\lambda_j+\lambda_m)t})}|_{t=0}=r_k^i \ell_k^i. \end{eqnarray*}
\end{proof}

Perhaps somewhat surprisingly, under the assumptions made, the mode-in-state participation factors are seen to agree with those of the linearized system.

\paragraph{Example}
Consider a nonlinear system whose linear part is from an example in \cite{abed}:
\[ \label{nl1}\dot{x}=\underbrace{\left[ \begin{array}{cc} a & b \\ 0 & d \end{array}\right]}_{A_1}x+\Psi(x),
\]
with $\Psi$ a polynomial of order $N \ge 2$. If $a \ne m \cdot d$ for any $m \in \NN$, then the eigenvalues of the matrix $A_1$ are nonresonant and, therefore, by the Poincar\'e's theorem there exists a formal transformation that transforms (\ref{nl1}) to
\[\label{poincare_lin1}\dot{z}=Az. \]
Furthermore, if $\lambda_1=a$ and $\lambda_2=d$
%herehere changed $\lambda_2$ to $\lambda_2=d$ in the preceding line
satisfy one of the conditions of the Poincar\'e-Siegel Theorem, then the transformation is analytic. In both cases, the mode-in-state participation factors of (\ref{nl1}) are locally equal to the mode-in-state participation factors of the linear system (\ref{poincare_lin1}).

A similar result holds for the following nonlinear system, whose linear part is from another example of \cite{abed}:
\[ \label{nl2}\dot{x}=\underbrace{\left[ \begin{array}{cc} 1 & 1 \\ -d & -d \end{array}\right]}_{A_1}x+\Psi(x),
\]
with $d \ne 1$ (nonresonance condition) and $\Psi$ is a polynomial of order $N \ge 2$. 

If the eigenvalues are resonant, and the origin is hyperbolic, we can still say something on the mode-in-state participation factors.

\begin{theorem} \cite{perko} If the origin is a hyperbolic point then there exists a homeomorphism that transforms the nonlinear ODE (\ref{nlsys}) into the linear ODE (\ref{linsys}). In this case, the mode-in-state participation factors of  (\ref{nlsys}) are the same as those of the linearized system $\dot{x}=Ax$.
\end{theorem}

\begin{proof}
First, we normalize $A$ using the change of coordinates (\ref{chancor1}) where $V=[r^1\, r^2 \, \cdots r^n]$ represents the matrix of right eigenvectors of $A$. Under the homeomorophism in the Hartman-Grobman theorem the ODE (\ref{nlsys}) becomes
\[\dot{z}=\Lambda z \]
The proof regarding mode-in-state participation factors comes directly from applying the result for the linear case in Section 3.

\end{proof}

\paragraph{Example}\cite{perko} Consider the system \begin{eqnarray}\label{eqn1} \dot{y}&=&-y \\ \label{eqn1a} \dot{z}&=&z+y^2 \end{eqnarray}
It can be shown \cite{perko} that with the homeomorphism
\[\phi(y,z)=\left[ \begin{array}{c}y \\ z+\frac{y^2}{3} \end{array}\right] \]
the solution of (\ref{eqn1})-(\ref{eqn1a}) is homeomorphic to the solution of 
\begin{eqnarray}\label{eqn2} \dot{y}&=&-y \\ \dot{z}&=&z \end{eqnarray}
and, therefore, the mode-in-state participation factors of the nonlinear system are the same as those of the linearized system.

%Here \[\phi^{-1}(y,z)=\left[ \begin{array}{c}y \\ z-\frac{y^2}{3} \end{array}\right] \]

\section{Conclusion}
There is a dichotomy in modal participation for linear systems. Hence we expect a similar dichotomy for nonlinear systems. Participation of modes in states is relatively easy to evaluate using averaging over an uncertain set of initial conditions assuming symmetric uncertainty.  Somewhat surprisingly, the mode-in-state participation formulas under these circumstances were found to be the same for a nonlinear system as for its linearization, assuming the nonresonance condition. Participation of states in modes for nonlinear systems is an open question, and its distinction from mode-in-state participation factors is part of the  dichotomy in modal participation seen in the linear case. Besides calculation of state-in-mode participation factors, some other issues that could be considered in future work are: computing modal participation factors for nonlinear systems from data; using the Frobenius-Perron operator to compute these measures. {Another possible extension is to use the recently introduced ``nonlinear eigenvalues'' and ``nonlinear eigenvectors'' for nonlinear systems \cite{kawano,padoan} to introduce possibly ``more nonlinear'' notions of modal participation factors for nonlinear systems.}
{Moreover, as mentioned in Section 1, bridging between modal participation concepts that have been proposed and used in different engineering fields would be worthwhile.}

%; the equivalence of ODEs and SDEs as far as modal participation is concerned, and using SDEs to compute modal participation factors; and the possibility of using a homeomorphism instead of a diffeormorphism according the the Hartman-Grobman Theorem.

%%%%%%%%%%%%%%%%%%%%%%%%%%%%%%%%%%%%%%%%%%%%%%%%%%%%%%%%%%%%%%%%%%%%%%%%%%%%%%%%
\section{ACKNOWLEDGMENTS}

BH thanks the European Commission and the
Scientific and the Technological Research Council of Turkey (Tubitak)
for financial support received through Marie Curie Fellowships. EHA thanks the US Air Force Office of Scientific Research for partial support under grant \#FA9550-09-1-0538.

%%%%%%%%%%%%%%%%%%%%%%%%%%%%%%%%%%%%%%%%%%%%%%%%%%%%%%%%%%%%%%%%%%%%%%%%%%%%%%%%


\begin{thebibliography}{99}\setlength{\itemsep}{-0.3mm}

\bibitem{arnold} V.I. Arnold, Geometrical Methods in the Theory of Ordinary
Differential Equations (1983) Springer.

\bibitem{AbedLindsayHashlamoun2000}
E.H. Abed, D.~Lindsay, and W.A. Hashlamoun, On participation factors for linear systems, Automatica, 36 (10) (2000) 1489-1496.

\bibitem{verghese_sma2013} {L. Rouco, F.L. Pagola, G.C. Verghese and I.J. P\'{e}rez-Arriaga, Selective modal analysis, pp. 199-258 in J.H. Chow, Ed., Power System Coherency and Model Reduction, Springer, New York, 2013.}

\bibitem{cenedese2016} {A. Cenedese, M. Fagherazzi and P. Bettini, A novel application of Selective Modal Analysis to large-scale electromagnetic devices, IEEE Trans. Magnetics, 52 (3) (2016), paper no. 7203304.}

\bibitem{setiada2018} {H. Setiadia, A.U. Krismantoa, N. Mithulananthana and M.J. Hossain, Modal interaction of power systems with high penetration of renewable
energy and BES systems, Internat. J. Electrical Power and Energy Systems, 97 (2018) 385-395.}

\bibitem{kundur1992} {B. Gao, G.K. Morison and P. Kundur, Voltage stability evaluation using
modal analysis, IEEE Trans. Power Syst., 7 (4) (1992), 1529-1542.}
\bibitem{bh_abed_cdc} B. Hamzi and E. Abed, Local mode-in-state participation factors for nonlinear systems, Proc. of  53rd  IEEE Annual Conference on Decision and Control (CDC),pp. 43-48, 2014.

\bibitem{abed} W.A. Hashlamoun, M.A. Hassouneh and E.H. Abed,  New results on modal participation factors: Revealing a previously unknown dichotomy, IEEE Trans.  Automatic Control,  54 (7) (2009) 1439-1449.

\bibitem{hill} {Y. Song, D.J. Hill and T. Liu, State-in-mode analysis of the power flow Jacobian for static voltage stability, Internat. J. Electrical Power and Energy Systems 105 (2019) 671-678.}

%\bibitem{lasota}   Lasota, A. and M. C. Mackey (2008). Probabilistic Properties of Deterministic Systems. Cambridge University Press.

\bibitem{vittal}  G. Jang, V. Vittal and W. Kliemann, Effect of nonlinear modal interaction on control performance: use of normal forms technique in control design. I. General theory and procedure, IEEE Trans. Power Systems (13) (2) (1998) 401-407.

\bibitem{kawano} Y. Kawano and T. Ohtsuka, PBH tests for nonlinear systems, Automatica, 80 (2017) 135-142.

\bibitem{lancaster} {P. Lancaster and M. Tismenetsky, The Theory of Matrices, Second Edition (1985) Academic Press, San Diego.}

\bibitem{prasad} {M.J.S. Lopez and J.V.R. Prasad, Estimation of modal participation factors of linear time periodic systems using linear time invariant approximations, Journal of the American Helicopter Society 61 (2016) paper 045001.}

\bibitem{chopra96} {A.K. Chopra, Modal analysis of linear dynamic systems: Physical interpretation, J. Structural Engineering, 122 (1996) 517-527.}

\bibitem{yakout1994} {Y. Mansour, W. Xu, F. Alvarado and C. Rinzin, SVC Placement using critical modes of voltage stability, IEEE Transactions on Power Systems 9 (2) (1994) 757-763.}

\bibitem{netto} {M. Netto, Y. Susuki and L. Mili, Data-driven participation factors for nonlinear systems based on Koopman mode decomposition, IEEE Control
Systems Society Letters (to appear).}

\bibitem{ramos} {R.A. Ramos, A.G.M. Moraco,  T.C.C. Fernandes and R.V. de Oliveira, Application of extended participation factors to detect voltage fluctuations in distributed generation systems, Proc. IEEE Power and Energy Society General Meeting (2011) 1-6.}

\bibitem{moraco} {A.G.M. Moraco,  T.C.C. Fernandes, G.S. Garcia and R.A. Ramos, Statistical analysis for the detection of voltage fluctuations in distributed synchronous generation using extended participations factors, Proc. IEEE International Conf. on Industrial Technology (2012) 457 - 462.}

\bibitem{padoan} 
A. Padoan and A. Astolfi, ``Eigenvalues'' and ``poles'' of a nonlinear system: a geometric approach, Proc. 56th IEEE Conf. Decision and Control (2017) 2575-2580.

\bibitem{park} {H.S. Park and J.H. Kwon, Optimal drift design model for multi-story buildings subjected to dynamic lateral forces, The Structural Design of Tall and Special Buildings 12 (2003) 317?333.}


\bibitem{Perez-Arriaga1}
I.J. P\'erez-Arriaga, G.C. Verghese and F.C. Schweppe, Selective
modal analysis with applications to electric power systems, {P}art
{I}: {H}euristic introduction, IEEE Trans. Power
Apparatus and Systems 101 (9) (1982) 3117-3125.
  
  \bibitem{Perez-Arriaga2}
G.C. Verghese, I.J. P\'erez-Arriaga and F.C. Schweppe, Selective
modal analysis with applications to electric power systems, {P}art
{II}: {T}he dynamic stability problem, IEEE Trans.
Power Apparatus and
  Systems, 101 (9) (1982) 3126-3134.
  
\bibitem{perko} L. Perko, Differential Equations and Dynamical Systems, 3rd Ed., Springer, New York,  2006.
  
\bibitem{powerelectronics} N. Teerakawanich, P. Evans and C.M. Johnson, Oscillation analysis of an Active Gate Control circuit for series connected IGBTs, 15th International Power Electronics and Motion Control Conference, EPE-PEMC 2012 ECCE Europe, Novi Sad, Serbia (2012) paper DS1a.5.
 
\end{thebibliography}
\end{document}